%% file: 2025_CDC.tex
\RequirePackage[OT1]{fontenc}
\documentclass[letterpaper, 10 pt, conference]{ieeeconf}
\IEEEoverridecommandlockouts
\usepackage{amssymb, amsfonts, mathtools, xargs, tensor, units, cite, stmaryrd, mathrsfs, algpseudocode, graphicx, xcolor, amsmath, listings, algorithm, tikz, pgfplots, cancel, multicol, geometry, multirow, booktabs,subcaption}

\usepackage{comment}
\usepackage{cancel}
\usepackage{pgfplots}
\usepackage{pgfplotstable}
\usepackage{cancel}
\usepackage{amsmath}
\usepackage{algorithm}
\usepackage{caption}
\pgfplotsset{compat=newest}

\pgfplotsset{compat=1.17}
\usepgfplotslibrary{colormaps}
\definecolor{customgreen}{rgb}{0.0, 0.6, 0.0}

\pgfplotsset{compat=1.18}
\usetikzlibrary{3d, calc}
\pgfplotsset{compat=newest}

\newtheorem{assumption}{Assumption}

\newtheorem{theorem}{Theorem}

\definecolor{color9}{RGB}{0,139,139}    % Teal

\allowdisplaybreaks
\title{Improved Dwell-times for Switched Nonlinear Systems using Memory Regression Extension}
\author{
Muzaffar Qureshi$^{1}$, Tochukwu Elijah Ogri$^{1}$, Humberto Ramos$^{1}$,\\
Wanjiku A. Makumi$^{2}$, Zachary I. Bell$^{2}$, Rushikesh Kamalapurkar$^{1}$% <-this % stops a space
\thanks{This research was supported in part by the Air Force Research Laboratories under contract numbers FA8651-24-1-0019 and FA8651-23-1-0006 and the Office of Naval Research under contract number N00014-21-1-2481. Any opinions, findings, or recommendations in this article are those of the author(s), and do not necessarily reflect the views of the sponsoring agencies.}% 
\thanks{$^{1}$ Department of Mechanical and Aerospace Engineering, University of Florida, email: {\tt\footnotesize \{muzaffar.qureshi, tochukwu.ogri, jramoszuniga, rkamalapurkar\} @ufl.edu}.}%
\thanks{$^{2}$ Air Force Research Laboratories, Florida, USA, email: {
\tt \footnotesize \{zachary.bell.10, wanjiku.makumi\} @us.af.mil.}}}

\begin{document}
\maketitle
\pagestyle{empty}

\begin{abstract}
This paper presents a switched systems approach for extending the dwell-time of an autonomous agent during GPS-denied operation by leveraging memory regressor extension (MRE) techniques.  To maintain accurate trajectory tracking despite unknown dynamics and environmental disturbances, the agent periodically acquires access to GPS, allowing it to correct accumulated state estimation errors. The motivation for this work arises from the limitations of existing switched system approaches, where increasing estimation errors during GPS-denied intervals and overly conservative dwell-time conditions restrict the operational efficiency of the agent. By leveraging MRE techniques during GPS-available intervals, the developed method refines the estimates of unknown system parameters, thereby enabling longer and more reliable operation in GPS-denied environments. A Lyapunov-based switched-system stability analysis establishes that improved parameter estimates obtained through concurrent learning allow extended operation in GPS-denied intervals without compromising closed-loop system stability. Simulation results validate the theoretical findings, demonstrating dwell-time extensions and enhanced trajectory tracking performance.
\end{abstract}

\section{Introduction}
Autonomous agents are frequently deployed in challenging environments where reliable state feedback from GPS is unavailable or intermittently disrupted. Such scenarios arise in applications including underwater navigation, subterranean exploration, military surveillance, and the mapping of hostile radar fields, where agents must either remain undetected or operate with limited sensing capabilities \cite{muzaffar.Wynn.Veerle.2014, muzaffar.Cathrin.Helen.ea2019, muzaffar.Qureshi.Ogri.ea2024}. In such scenarios, agents predominantly rely on inertial navigation systems (INS) and relative sensing techniques for state estimation. However, these methods are inherently prone to error accumulation over time, leading to drift and progressively degraded localization accuracy. \cite{muzaffar.Chowdhary.Johnson.ea2013, muzaffar.Leishman.McLain.ea2013}.

Visual or LiDAR-based simultaneous localization and mapping (SLAM) is a widely adopted alternative for inertial navigation. SLAM techniques leverage environmental feature recognition to mitigate localization drift through loop closures detecting and revisiting previously encountered locations~\cite{Muzaffar.Netter.Francheschini.2002, muzaffar.Stephan.2012, muzaffar.Whyte.Bailey2006}. While effective in structured environments, SLAM inherently depends on the presence of distinctive and persistent features. In feature-sparse settings, such as underwater domains, planetary surfaces, or environments affected by adverse weather conditions, SLAM performance significantly deteriorates, thereby compromising the localization accuracy of the agent ~\cite{muzaffar.Yin.Wang.ea2021}.

To address these limitations, recent research has developed a switched-systems framework in which agents periodically access GPS to correct accumulated localization errors before resuming operation in a GPS-denied setting (see \cite{muzaffar.Chen.Yuan.ea2019, muzaffar.Greene.Sakha.ea2024, muzaffar.Qureshi.Ogri.ea2025a}). The intermittent use of GPS ensures that localization errors remain bounded at all times. 

In~\cite{muzaffar.Chen.Yuan.ea2019}, the authors derived the dwell-time conditions using Lyapunov-based stability analysis, establishing the minimum and maximum allowable time durations for GPS-available and GPS-denied intervals to ensure that state estimation errors remain within prescribed safety bounds. However, these dwell-time conditions are inherently conservative and do not leverage state information acquired during GPS-enabled intervals, resulting in frequent use of GPS to maintain stability.

In this paper, the dwell-time conditions are relaxed using memory regressor extension (MRE) \cite{SCC.Bell.Chen.ea2017,muzaffar.Sethares.Anderson.ea1989,muzaffar.Slotine.Li1989,SCC.Bell.Nezvadovitz.ea2018} techniques while still ensuring that estimation errors remain within predefined safety bounds. By leveraging MREs, we make the error bounds less conservative over time, allowing for more efficient state estimation and trajectory tracking.

A switched-systems framework is developed to extend the permissible dwell-time in GPS-denied intervals using models learned via MREs when GPS is available. The developed approach systematically exploits measurements collected during intermitted periods of GPS availability to refine the system model. By improving model accuracy, MRE techniques reduce the growth of localization errors in subsequent GPS-denied intervals. Compared to prior methods, the developed framework relaxes conservative dwell-time conditions while preserving Lyapunov-based stability guarantees during operation in GPS-denied intervals.

The primary contribution of this paper lies in integrating parameter estimation strategies into a switched-system framework to derive improved dwell-time conditions as functions of modeling errors. This integration allows the user to increase the duration of GPS-denied navigation intervals as learned models become more accurate while maintaining predefined error bounds. The effectiveness of the developed switched-system approach, relative to existing methods, is established through a Lyapunov-based stability analysis and validated via numerical simulations.

\section{Problem Formulation}
Consider an autonomous agent tasked with tracking a desired trajectory \( x_d \) within an environment where continuous state feedback is unavailable. The control-affine system describes the dynamics of the agent
\begin{equation}\label{eq:dynamics}
\dot{x}(t) = f(t, x) + u(t) + Y(t, x)\theta + d(t, x),
\end{equation}
where \( x \in \mathbb{R}^n \) denotes the state of the agent, \( u : \mathbb{R}_{\geq 0} \rightarrow \mathbb{R}^n \) is the control input, \( f : \mathbb{R}_{\geq 0} \times \mathbb{R}^n \rightarrow \mathbb{R}^n \) represents known part of the drift dynamics, and \( Y : \mathbb{R}_{\geq 0} \times \mathbb{R}^n \rightarrow \mathbb{R}^{n \times p} \) is a known regressor matrix associated with the unknown part of th dynamics. The parameter vector \( \theta \in \mathbb{R}^p \) is unknown, and \( d : \mathbb{R}_{\geq 0} \times \mathbb{R}^n \rightarrow \mathbb{R}^n \) represents unmodeled residual dynamics. To ensure the existence and uniqueness of solutions to~\eqref{eq:dynamics}, the following assumptions are required
\begin{assumption}\label{assump:lipschitz}
The functions \( f \) and  \( Y \) are piecewise continuous with respect to $t$ and are locally Lipschitz continuous with respect to \( x \) uniformly in $t$. Specifically, given any compact set $\mathcal{X} \subset \mathbb{R}^n$, there exist constants \( L_f, L_y > 0 \), such that 
\begin{align}
\|f(t,x)-f(t,\hat{x})\| &\leq L_f \|x-\hat{x}\|,\\
\|Y(t,x)-Y(t,\hat{x})\| &\leq L_y\| x-\hat{x}\|,
\end{align}
for all \( x,\hat{x} \in \mathcal{X} \), and for all \( t \in \mathbb{R}_{\geq 0} \).
Finally, \( f(t,0) = 0 \), \( Y(t,0) = 0 \) for all \(\forall t \in \mathbb{R}_{\geq 0}\).
\end{assumption}
Under Assumption~\ref{assump:lipschitz}, the regressor matrix \( Y \) is bounded on $\mathcal{X}$ by a positive constant \( \bar{Y} > 0 \) such that
$\| Y(t, x) \| \leq \bar{Y}$, for all $x \in \mathcal{X}$ and for all $t \in \mathbb{R}_{\geq 0}$.

\begin{assumption}\label{assump:bounded_disturbance}
The unmodeled dynamics \( d \) are bounded such that, $ \forall (t,x) \in \mathbb{R}_{\geq0} \times \mathbb{R}^n, \| d(t,x) \| \leq \bar{d},$ where \( \bar{d} \in \mathbb{R}_{\geq0} \).
\end{assumption}
\begin{assumption}
The desired trajectory $x_d$ is bounded by a known positive constant $\overline{x}_d \in \mathbb{R}_{> 0}$, such that $\| x_d (t)\| \leq \overline{x}_d$, $\forall t \in \mathbb{R}_{\geq 0}$.
\end{assumption}

The control objective is to track a desired trajectory while minimizing the need to measure the system state. In GPS-denied intervals, the controller relies on model-based state estimates. Due to modeling errors included by the disturbance $d$ and the unknown parameters $\theta$, the state estimation error, and consequently, the tracking error grows in GPS-denied intervals. Since the growth rate of the error depends on the parameter estimation error, the time the agent can spend without access to GPS while maintaining prescribed tracking error bounds can be increased as the parameters are learned. In this paper, we developed a scheduling technique that enables the agent to extend its operational time without GPS access while maintaining stability and tracking performance. In particular, a memory regressor extension technique is integrated into a switched-system framework to enhance the performance of an autonomous agent operating without GPS. 

\section{Overview}
During GPS-available intervals, the agent leverages available state measurements to estimate the unknown model parameters \( \theta \), and during GPS-denied intervals, the agent uses the parameter estimates to estimate the system state.

The state estimation error \( e_1 : \mathbb{R}_{\geq 0} \rightarrow \mathbb{R}^n \) is defined as
\begin{equation}\label{eq:e1}
e_1(t) \coloneqq x(t) - \hat{x}(t),
\end{equation}
where \( \hat{x}(t) \in \mathbb{R}^n \) is the state estimate.
The trajectory tracking error \( e_2 : \mathbb{R}_{\geq 0} \rightarrow \mathbb{R}^n \) is defined as
\begin{equation}\label{eq:e2}
e_2(t) \coloneqq \hat{x}(t) -{x}_{d}(t).
\end{equation}
Let the concatenated error vector be defined as \( e(t) \coloneqq \begin{bmatrix} e_1(t)^\top & e_2(t)^\top \end{bmatrix}^\top \in \mathbb{R}^{2n} \). Note that \( e_1 \) is only measurable when the agent has access to GPS while \( e_2 \) is always measurable. Let \( V : \mathbb{R}^{2n} \to \mathbb{R} \) be a candidate Lyapunov function defined as
\begin{equation}\label{eq:lyapunov}
V(e) \coloneqq \frac{1}{2} e_1^\top e_1 + \frac{1}{2} e_2^\top e_2.
\end{equation}
The calculation of dwell-times depends on two user-defined bounds on $V$, an upper bound \( V_u \in \mathbb{R}_{>0} \text{ and a lower bound } V_l \in \mathbb{R}_{>0} \). These bounds are used to find an upper bound on the time the agent can spend without GPS and a lower bound on the time the agent must spend with GPS to ensure that \( V_l \leq V(e) \leq V_u \) for all \( t \). The idea is to use a Lyapunov-based argument to bound the rate of increase of $V$ when the agent is operating without state feedback and bound the rate of decrease of $V$  when the agent is operating with state feedback. 

% The upper bound $V_{u}$ is relevant to the segment of the trajectory in $\mathcal{X}_{u}$, where the estimation error grows and $V$ tends to increase. By bounding $V$ from above, we also bound the maximum allowable time the agent can stay in $\mathcal{X}_{u}$. On the other hand, the lower bound $V_{l}$ is relevant to $\mathcal{X}_{a}$ which defines the minimum time the agent must $\mathcal{X}_{a}$ before exiting into $\mathcal{X}_{u}$.
To develop the switched systems framework, we assume that the GPS is available at the initial time \( t_0^a \). Let  \( t_0^u \) denote the first time instance when GPS becomes unavailable. Subsequent GPS availability and unavailability instances are indexed by \( \sigma \), and denoted by \( t_\sigma^a \) and \( t_\sigma^u \), respectively. During the time interval \( t \in \left[t_\sigma^a, t_\sigma^u\right) \), GPS is available, while for \( t \in \left[t_\sigma^u, t_{\sigma+1}^a\right) \), GPS is not available. Let \( \Delta t^a_\sigma \coloneqq t^u_\sigma - t^a_\sigma \) and \( \Delta t^u_\sigma \coloneqq t^a_{\sigma+1} - t^u_\sigma \) denote the durations of the GPS-available and the GPS-denied intervals, respectively. Let $\mathcal{N}^{o} \subseteq \mathbb{N}$ denote the number of switches between GPS availability and unavailability.

In this paper, the agent leverages state measurements obtained during \( [t_{\sigma}^a,t_{\sigma}^u) \) to refine estimates of the unknown parameter vector \( \theta \). The idea is to utilize these improved estimates to progressively extend the permissible dwell-time $\Delta t^u_\sigma$ over time while ensuring that state estimation errors remain within the same user-defined bounds. The following section presents the MRE formulation that enables robust parameter estimation using state data during \( [t_{\sigma}^a,t_{\sigma}^u)\).

\section{Memory Regressor Extension}  
To estimate the parameters, we utilize memory-based adaptive control. To facilitate the estimation of the unknown parameter vector \( \theta \), we express \eqref{eq:dynamics} as
\begin{equation}\label{eq:regressor_form}
\mathcal{U}_{f\sigma}(t) = \mathcal{Y}_{f\sigma}(t)\theta + \Xi_{\sigma}(t), \quad \sigma \in \mathcal{N}^{o},
\end{equation}
where \( \mathcal{U}_{f\sigma} : \mathbb{R}_{\geq 0} \rightarrow \mathbb{R}^n \), \( \mathcal{Y}_{f\sigma} : \mathbb{R}_{\geq 0} \rightarrow \mathbb{R}^{n \times p} \), and \( \Xi_{f\sigma} : \mathbb{R}_{\geq 0} \rightarrow \mathbb{R}^n \) represent the filtered state and control inputs, the filtered regressor matrix, and the filtered residual modeling error, respectively. The signals \( \mathcal{U}_{f\sigma}\) and \( \mathcal{Y}_{f\sigma} \) are computed using measurements of the system state \( x \) and the input \( u \). The residual term \( \Xi_{f\sigma}(t) \) satisfies \( \Xi_{f\sigma}(t) = \mathcal{O}(\bar{d}) \) for all \( t \geq 0 \), where \( \bar{d} \) is an upper bound on the unmodeled dynamics\footnote{The notation \( \Xi_{f\sigma}(t) = \mathcal{O}(\bar{d}) \) means that there exist a constant \( c \in \mathbb{R}_{>0} \) such that \( \|\Xi_{f\sigma}(t)\| \leq c \, \bar{d} \) for all \( t \in \mathbb{R}_{\geq 0} \).}.

Several approaches can be employed to compute the filtered signals in~\eqref{eq:regressor_form}. One common technique is using regressor filtering \cite{muzaffar.Nguyen2018}, which involves convolving the signals with an exponentially decaying filter \( h(t) = \beta e^{-\beta t} \), where \( \beta > 0 \), yielding
{\medmuskip=2mu\thinmuskip=2mu\thickmuskip=2mu\begin{equation}
    \mathcal{U}_{f\sigma}(t) = \begin{cases}
        [h * \left(\dot{x}(\cdot)-f(\cdot,x(\cdot))-u(\cdot)\right)]{(t)}, & t \in (t_{\sigma}^{a}, \; t_{\sigma}^{u}], \\
        0_{n \times 1}, & t \in (t_{\sigma}^{u}, t_{\sigma + 1}^{a}],
    \end{cases}
\end{equation}}
\begin{equation} \mathcal{Y}_{f\sigma}(t) =  \begin{cases}
        [h * Y(\cdot,x(\cdot))]{(t)}, & t \in (t_{\sigma}^{a}, \; t_{\sigma}^{u}], \\
        0_{n \times p}, & t \in (t_{\sigma}^{u}, t_{\sigma + 1}^{a}],
    \end{cases}
\end{equation}
and
\begin{equation}
    \Xi_{f\sigma}(t)  = \begin{cases}
         [h * d(\cdot,x(\cdot))]{(t)}, & t \in (t_{\sigma}^{a}, \; t_{\sigma}^{u}], \\
        0_{n \times 1}, & t \in (t_{\sigma}^{u}, t_{\sigma + 1}^{a}].
    \end{cases}
\end{equation}
Another MRE approach is windowed integration, where the signals are integrated over a finite time window of length \( \Delta t > 0 \). Let 
$\mathcal{I}_{u}(t) \coloneqq x(t) - x(t - \Delta t) - \int_{t-\Delta t}^{t} f(\tau, x(\tau)) - u(\tau) \, \mathrm{d}\tau$, $\mathcal{I}_{y}(t) \coloneqq \int_{t-\Delta t}^{t} Y(\tau, x(\tau)) \, \mathrm{d}\tau$, and $\mathcal{I}_{d}(t) \coloneqq \int_{t-\Delta t}^{t} d(\tau, x(\tau))\, \mathrm{d}\tau$ with $t > t_{\sigma}^{a} + \Delta t$. Then the signals  \(\mathcal{U}_{f\sigma}:\mathbb{R}_{\geq 0}\rightarrow \mathbb{R}^n\), \(\mathcal{Y}_{f\sigma}:\mathbb{R}_{\geq 0}\rightarrow\mathbb{R}^{n\times p}\), and \(\Xi_{f\sigma}:\mathbb{R}_{\geq 0}\rightarrow\mathbb{R}^n\) can be defined as
\begin{equation}\label{eq:filteredSignal1}
    \mathcal{U}_{f\sigma}(t) = \begin{cases}
        \mathcal{I}_{u}(t) - \mathcal{I}_{u}(t_{\sigma}^{a}), & t \in (t_{\sigma}^{a}, t_{\sigma}^{a} + \Delta t],\\
        \mathcal{I}_{u}(t) - \mathcal{I}_{u}(t - \Delta t), & t \in (t_{\sigma}^{a} + \Delta t, t_{\sigma}^{u}], \\
        0_{n \times 1}, & t \in (t_{\sigma}^{u}, t_{\sigma + 1}^{a}],
    \end{cases}
\end{equation}
\begin{equation}\label{eq:filteredSignal2}
    \mathcal{Y}_{f\sigma}(t) = \begin{cases}
        \mathcal{I}_{y}(t) - \mathcal{I}_{y}(t_{\sigma}^{a}), & t \in (t_{\sigma}^{a}, t_{\sigma}^{a} + \Delta t],\\
        \mathcal{I}_{y}(t) - \mathcal{I}_{y}(t - \Delta t), & t \in (t_{\sigma}^{a} + \Delta t, t_{\sigma}^{u}], \\
        0_{n \times p}, & t \in (t_{\sigma}^{u}, t_{\sigma + 1}^{a}],
    \end{cases}
\end{equation}
and 
\begin{equation}\label{eq:filteredSignal3}
    \Xi_{f\sigma}(t) = \begin{cases}
        \mathcal{I}_{d}(t) - \mathcal{I}_{d}(t_{\sigma}^{a}), & t \in (t_{\sigma}^{a}, t_{\sigma}^{a} + \Delta t],\\
        \mathcal{I}_{d}(t) - \mathcal{I}_{d}(t - \Delta t), & t \in (t_{\sigma}^{a} + \Delta t, t_{\sigma}^{u}], \\
        0_{n \times 1}, & t \in (t_{\sigma}^{u}, t_{\sigma + 1}^{a}],
    \end{cases}
\end{equation}
respectively.

Once $\mathcal{U}_{{f \sigma}}$ and $\mathcal{Y}_{{f \sigma}}$ are obtained, a MRE technique, such
as concurrent learning (CL) \cite{SCC.Parikh.Kamalapurkar.ea2019} is used to compute \( \mathcal{U}_{\sigma}, \mathcal{Y}_{\sigma}\) and $\Xi_{\sigma}$, such that
\begin{equation}
    \mathcal{U}_{\sigma}(t)= \mathcal{Y}_{\sigma}(t)\theta+\Xi_{\sigma}(t),
\end{equation} 
where \( \mathcal{U}_{\sigma}: \mathbb{R}_{\geq 0}  \mapsto  \mathbb{R}^{p} \) and \( \mathcal{Y}_{\sigma}: \mathbb{R}_{\geq 0} \mapsto \mathbb{R}^{p \times p} \) can be computed using measurements of \( x \) and \( u \) provided $\Xi_{f\sigma}(t)=\mathcal{O}(\bar{d})$. In concurrent learning, a set of discrete time instances \( \{t_i\}_{i=1}^N \subseteq (t_\sigma^a,t_\sigma^u)\), with \( N > p \), is selected to store the values of  $\mathcal{U}_{f \sigma}$ and $\mathcal{Y}_{f \sigma}$. Using these time-stamped samples, \( \mathcal{U}_{\sigma}, \mathcal{Y}_{\sigma} \) and $\Xi_{\sigma}$ can be expressed for $t \in (t_\sigma^a,t_\sigma^u)$ as
\begin{align}
\mathcal{U}_\sigma (t) & \coloneqq \sum_{i=1}^N \frac{\mathcal{Y}_{f \sigma}^{\top}(t_i(t)) \mathcal{U}_{f \sigma}(t_i(t))}{1 + \| \mathcal{Y}_{f \sigma}(t_i(t) \|^2}, \ \mathcal{U}_{\sigma}(t_{\sigma}^a) =0_{p\times 1} \label{eq:concurrent_F}, \\
\mathcal{Y}_\sigma (t) & \coloneqq \sum_{i=1}^N \frac{\mathcal{Y}_{f \sigma}^{\top}(t_i(t)) \mathcal{Y}_{f \sigma}(t_i(t))}{1 + \| \mathcal{Y}_{f \sigma}(t_i(t)) \|^2}, \ \mathcal{Y}_{\sigma}(t_{\sigma}^a) =0_{p \times p},\label{eq:concurrent_Y} \\
\Xi_\sigma (t) & \coloneqq \sum_{i=1}^N \frac{\mathcal{Y}_{f \sigma}^{\top}(t_i(t)) \Xi_{\sigma}(t_i(t))}{1 + \| \mathcal{Y}_{f \sigma}(t_i(t)) \|^2},\ \Xi_{\sigma}(t_{\sigma}^a) =0_{p \times 1}, \label{eq:concurrent_Xi}
\end{align}

Another approach to compute $\mathcal{U}_{\sigma}$ and $\mathcal{Y}_{\sigma}$ is using exponentially weighted integrals \cite{muzaffar.Garg.BasuRoy.ea2024}, where a forgetting factor \( \alpha > 0 \) is used to prioritize recent data while retaining informative contributions from earlier measurements. That is, for $t \in [t_{\sigma}^a,t_{\sigma}^u)$,
\begin{equation}
    \mathcal{U}_\sigma(t) = \int_{t_{\sigma}^{a}}^{t} e^{-\alpha \tau}\left( \frac{\mathcal{Y}_{f\sigma}^{\top}(\tau) \mathcal{U}_{f\sigma}(\tau)}{1+\|\mathcal{Y}_{f\sigma}(\tau)\|^2}\right) \mathrm{d}\tau,
\end{equation}
\begin{equation}
    \mathcal{Y}_\sigma(t) = \int_{t_{\sigma}^{a}}^{t} e^{-\alpha \tau} \left(\frac{\mathcal{Y}_{f \sigma}^{\top}(\tau) \mathcal{Y}_{f \sigma}(\tau)}{1+\|\mathcal{Y}_{f \sigma}(\tau)\|^2 }\right)\mathrm{d}\tau,
\end{equation}
and
\begin{equation}
    \Xi_\sigma(t) = \int_{t_{\sigma}^{a}}^{t} e^{-\alpha \tau} \left(\frac{\mathcal{Y}_{f \sigma}^{\top}(\tau) \Xi_{{f \sigma}}(\tau)}{1+\|\mathcal{Y}_{f \sigma}(\tau)\|^2}\right)\, \mathrm{d}\tau.
\end{equation}
 The developed framework can also incorporate other MRE techniques, such as those developed in \cite{muzaffar.Ortega.Nikiforov.ea2020,SCC.Kreisselmeier1997,SCC.Bell.Parikh.ea2016, SCC.Ogri.Bell.ea2023}.  Convergence of the parameter estimation error to a neighborhood of the origin follows if the regressor $\mathcal{Y}_{\sigma} $ is sufficiently exciting, as formalized in the following assumption. 
\begin{assumption}\label{assump:fe}
    The regressor $\mathcal{Y}_{\sigma}$ is uniformly sufficiently exciting i.e. for all $\sigma \in \mathcal{N}^{o}$, there exist constants $\lambda_{y} \in \mathbb{R}_{> 0}$ and $T_{\sigma} \in [t_{\sigma}^{a}, t_{\sigma}^{u})$ such that for $T_\sigma 
    \leq t \leq t_{\sigma+1}^a$, and for all initial conditions $e_1(t_{\sigma}^a),e_2(t_{\sigma}^a)$, and $\tilde{\theta}(t_{\sigma}^a)$,
\begin{equation}\label{eq:regressorBound}
    \lambda_{\min}(\mathcal{Y}_{\sigma})  > \lambda_{y}.
\end{equation}
\end{assumption}

Based on the subsequent stability analysis in Section~\ref{section:theta_stability} and under Assumption~\ref{assump:fe}, an adaptive update law is designed as
\begin{equation}\label{eq:thetaUpdateSwitching}
\dot{\hat{\theta}}(t) =
\begin{cases}
k_{\theta} \Gamma\left(\mathcal{U}_{\sigma}(t) - \mathcal{Y}_{\sigma}(t) \hat{\theta}(t) \right), & \quad t \geq T_{\sigma}, \\
0_{p\times 1},& \quad t < T_{\sigma},
\end{cases}
\end{equation}
where \( k_{\theta} > 0 \) is a constant scalar, \(\Gamma \in \mathbb{R}^{p\times p}\) is a constant positive definite matrix, and $\hat{\theta} \in \mathbb{R}^{p}$ denotes the estimate of $\theta$. Let the parameter estimation error be defined as 
\begin{equation}\label{eq:paramE}
    \tilde{\theta}(t) \coloneqq \theta - \hat{\theta}(t).
\end{equation}
Differentiating \eqref{eq:paramE} and substituting the update law from \eqref{eq:thetaUpdateSwitching} yields the parameter estimation error dynamics
\begin{equation}\label{eq:parameterErrorDynamics2}
\dot{\tilde{\theta}}(t) =  \begin{cases} -k_{\theta}\Gamma\mathcal{Y}_{\sigma}(t)\tilde{\theta}(t) -k_{\theta}\Gamma\Xi_{\sigma}(t),& \quad t \geq T_{\sigma}, \\
0_{p \times 1},& \quad t < T_{\sigma}.
\end{cases}
\end{equation}

The next section analyzes the stability of the parameter estimation error system in \eqref{eq:parameterErrorDynamics2} and establishes uniform ultimate boundedness of its trajectories.

\section{Stability Analysis for CL Update Laws}\label{section:theta_stability}
The stability properties of the MRE-based observer developed in \eqref{eq:thetaUpdateSwitching} are summarized in the
following theorem.
\begin{theorem}
If Assumptions~\ref{assump:lipschitz}--\ref{assump:fe} are satisfied, then the adaptive update law defined in \eqref{eq:thetaUpdateSwitching} ensures global uniform ultimate boundedness of the parameter estimation error \(\tilde{\theta}\) during the GPS-available interval \( [t_\sigma^a, t_\sigma^u) \) for each $\sigma \in \mathcal{N}^{o}$.
\label{theorem:theorem1}
\end{theorem}
\begin{proof}
Consider the candidate Lyapunov function $V_{\sigma}: \mathbb{R}^{p} \to \mathbb{R}_{\geq 0}$ defined as 
\begin{equation}
V_\sigma(\tilde{\theta}) \coloneqq \frac{1}{2}\tilde{\theta}^\top\Gamma^{-1}\tilde{\theta},
\end{equation}
which satisfies the inequality $\underline{\Gamma} \|\theta\|^{2} \leq V_\sigma(\tilde{\theta}) \leq \overline{\Gamma}  \|\theta\|^{2}$ for $\sigma \in \mathcal{N}^{o}$, where the bounds $\underline{\Gamma} \coloneqq \lambda_{\min}\left(\Gamma^{-1}\right)$ and $\overline{\Gamma} \coloneqq \lambda_{\max}\left(\Gamma^{-1}\right)$ are positive constants.
Taking the Lie derivative of $V_{\sigma}$ along the flow of the parameter error dynamics \eqref{eq:parameterErrorDynamics2} yields
 \begin{equation}
\dot{V}_\sigma(t,\tilde{\theta}) = \begin{cases}- k_{\theta}\tilde{\theta}^{\top}\mathcal{Y}_{\sigma}\tilde{\theta} - k_{\theta}\tilde{\theta}^{\top}\Xi_{\sigma}, & t \geq T_{\sigma}, \\
0, & t < T_{\sigma}. 
\end{cases}
\end{equation}
Using Assumption~\ref{assump:bounded_disturbance} and Assumption~\ref{assump:fe}, the Lie derivative of the candidate Lyapunov function is bounded by
 \begin{equation}
\dot{V}_\sigma(t,\tilde{\theta}) \leq - k_{\theta}\lambda_{y}\|\tilde{\theta}\|^{2} +k_{\theta}k_{\xi}\overline{d}\|\tilde{\theta}\|,
\end{equation}
where $k_{\xi} \in \mathbb{R}_{>0}$ is a bound on $\Xi(t)$ such that $\|\Xi(t)\| \leq k_{\xi}\overline{d}$, $\forall t \in [t_\sigma^a, t_\sigma^u)$. Applying Young's inequality, and defining the constants $\rho \coloneqq \frac{k_{\theta}\lambda_{y}}{4}$ and $\varpi \coloneqq \frac{k_{\theta}k_{\xi}^{2}\overline{d}^{2}}{2\lambda_{y}}$, the derivative of $V_{\sigma}$ can be bounded as
\begin{equation}
    \dot{V}_\sigma(t,\tilde{\theta}) \leq -\rho\left(\|\tilde{\theta}\|^{2} - \frac{\varpi}{\rho}\right).
\end{equation}
Hence, the Lie derivative of the candidate Lyapunov function satisfies $\dot{V}_\sigma(\tilde{\theta}) \leq -\rho\|\tilde{\theta}\|^{2}$ for all $\|\tilde{\theta}\| > \sqrt{\frac{\varpi}{\rho}} > 0$. Hence, \cite[Theorem~4.18]{SCC.Khalil2002} can be invoked to conclude that the trajectories of \eqref{eq:parameterErrorDynamics2} are globally uniformly ultimately bounded. In particular, for every trajectory $\tilde{\theta}(t)$ with arbitrary initial condition, we have that $\limsup_{t \to \infty}\big\|\tilde{\theta}(t)\big\| \leq \sqrt{\frac{\overline{\Gamma}}{\underline{\Gamma}}\left(\frac{\varpi}{\rho}\right)}$, which completes the proof.
\end{proof}

During the interval \( [t_\sigma^a, t_\sigma^u) \) when GPS is available, the parameter estimation error trajectories satisfy
\begin{equation}\label{eq:theta_bar_norm}
\|\tilde{\theta}(t)\|
    \;\le\;
    \sqrt{\frac{\overline{\Gamma}}{\underline{\Gamma}}\left(
      \|\tilde{\theta}(t_{\sigma}^a)\|^2\,
      e^{-\rho t/\overline{\Gamma}}
      \;+\;
      \frac{\varpi}{\rho}\,
      \Bigl(1 - e^{-\rho t/\overline{\Gamma}}\Bigr)\right)
    },
\end{equation}
for all $t \in [t_{\sigma}^a,t_{\sigma}^u)$. However, in the interval \( [t_\sigma^u, t_{\sigma+1}^a) \) when the agent does not have access to GPS, the parameter estimates $\hat{\theta}$ are not updated, i.e., $ \dot{\hat{\theta}}(t)=0$, $\forall t\in[t_\sigma^u,t_{\sigma+1}^a).$
The parameter estimation error thus remains unchanged during the interval $[t_\sigma^u, t_{\sigma+1}^a)$, i.e.,
\begin{equation}\label{eq:theta_tilde_ineq}
\|\tilde{\theta}(t_{\sigma+1}^a)\|=\|\tilde{\theta}(t_\sigma^u)\|.
\end{equation}
Extending this analysis recursively across \( k \) intervals, the estimation error after \( k \) GPS-available intervals satisfies
\begin{equation}\label{eq:final_theta_bar}
\|\tilde{\theta}(t^u_{\sigma+k})\| \leq 
\sqrt{ 
\alpha^{k+1} \|\tilde{\theta}(t^a_\sigma)\|^2 \prod_{j=0}^k e^{-\frac{\rho}{\overline{\Gamma}} \Delta t_{\sigma+j}} + 
\gamma \cdot \Phi_k
},
\end{equation}
where $\Phi_k := \sum_{j=0}^k \alpha^{k-j} \prod_{\ell = j+1}^{k} e^{-\frac{\rho}{\overline{\Gamma}} \Delta t_{\sigma+\ell}} \left(1 \!\!-\!\! e^{-\frac{\rho}{\overline{\Gamma}} \Delta t_{\sigma+j}}\right)$, \( \alpha := \frac{\overline{\Gamma}}{\underline{\Gamma}} \), and \( \gamma := \frac{\varpi}{\rho} \).
From the inequality in \eqref{eq:final_theta_bar}, we can see that after some finite time \( t_r \), the parameter estimation error $\tilde{\theta}(t)$ remains confined to a ball of radius $\sqrt{\frac{\overline{\Gamma}}{\underline{\Gamma}}\left(\frac{\varpi}{\rho}\right)}$. Letting \( \bar{\mathcal{N}} \subset \mathcal{N}^o \) denote the subset of switching indices such that \( t_\sigma^u < t_r \), the parameter estimation errors satisfy
\begin{equation}\label{eq:recursive_product_subset}
\|\tilde{\theta}(t_{\sigma+k}^u)\| <\|\tilde{\theta}(t_\sigma^u)\|, \quad \forall \sigma \in \bar{\mathcal{N}}.
\end{equation} 
The subsequent section presents the state observer and controller design, which leverages estimates from the update law in \eqref{eq:thetaUpdateSwitching} to enable the agent to operate within the user-specified error bounds.

\section{Controller Design and Observer Update Laws}
In the absence of continuous state feedback, an observer is required to estimate the state of the agent. This section presents a switching observer along with a stabilizing controller that enables the agent to track a desired trajectory, regardless of GPS availability.
% \hr{Showing the steps to get the expressions you need is good. Briefly explaining WHY you are executing such steps is great. Make sure that your sections do not become a series of void steps. The reader and reviewer always appreciate that you tell them WHY you are taking them through the math that you are showing.}
Motivated by the observer design for switched systems presented in \cite{muzaffar.Chen.Yuan.ea2019}, the state estimates $\hat{x} \in \mathbb{R}^{n}$ are computed using an observer of the form
\begin{equation}\label{eq:state observer}
\dot{\hat{x}}_{\sigma} = 
\begin{cases}
    f(t, \hat{x}) + Y(t,\hat{x})\hat{\theta} + u(t) + v_{r,\sigma}, & \forall t \in \left[t_\sigma^a, t_{\sigma}^u\right), \\
    f(t, \hat{x}) + Y(t,\hat{x})\hat{\theta} + u(t), & \forall t \in [ t_\sigma^u, t_{\sigma+1}^a ),
\end{cases}
\end{equation}
with the sliding mode term $v_{r,\sigma}$ designed as  
\begin{equation}
v_r,\sigma = k_1 e_1 + (\bar{d} + \bar{Y} \tilde{\theta}_{\sigma})\, \text{sgn}(e_1),
\end{equation}  
where \( k_1 \in \mathbb{R}^{n \times n} \) is a positive definite control gain matrix, sgn$(\cdot)$ is the sign function, and $\tilde{\theta}_{\sigma}= \|\tilde{\theta}(t_\sigma^u)\|$ is the error bound at start of the GPS-denied interval derived in \eqref{eq:theta_bar_norm}. Note that since the observer is discontinuous, solutions of \eqref{eq:state observer} are interpreted in a generalized sense as Krasovskii solutions \cite{SCC.Shevitz.Paden1994}.

For notation brevity, the dependence of all functions on $t$ is omitted hereafter unless needed for clarity. Also, let $f_x \coloneqq f(t,x)$, $f_{\hat{x}} \coloneqq f(t,\hat{x})$, $Y_x \coloneqq Y(t,x)$, $Y_{\hat{x}} \coloneqq Y(t,\hat{x})$, and $d_x \coloneqq d(t,x)$. Using the state and parameter estimates from \eqref{eq:state observer} and \eqref{eq:thetaUpdateSwitching}, respectively, a stabilizing controller can be designed as
\begin{equation}
\label{eq:controller}
\scalebox{0.95}{$
u = 
\begin{cases}
    \dot{\bar{x}}_d - f_{\hat{x}} - k_2 e_2 - v_{r,\sigma} - Y_{\hat{x}}\hat{\theta}, & \forall t \in \left[t_\sigma^a, t_{\sigma}^u\right), \\
    \dot{\bar{x}}_d - f_{\hat{x}} - k_2 e_2- Y_{\hat{x}}\hat{\theta}, & \forall t \in [t_\sigma^u, t_{\sigma+1}^a ),
\end{cases}
$}
\end{equation}
where \( k_2 \in \mathbb{R}^{n \times n} \) is a positive definite tracking error gain matrix.  Substituting \eqref{eq:dynamics} and \eqref{eq:state observer} into the time derivative of \eqref{eq:e1} yields the state estimation error dynamics
{\medmuskip=2mu\thinmuskip=2mu\thickmuskip=2mu\begin{equation}\label{eq:e1_dot}\scalebox{0.93}{$
\dot{e}_1 =
\begin{cases}
    f_x + Y_x \theta+ d_x - f_{\hat{x}} - Y_{\hat{x}}\hat{\theta}- v_{r,\sigma} , & \forall t \in \left[t_\sigma^a, t_{\sigma}^u\right),\\
    f_x + Y_x\theta + d_x- f_{\hat{x}} - Y_{\hat{x}}\hat{\theta}, & \forall t \in \left[t_\sigma^u, t_{\sigma+1}^a\right).
\end{cases}$}
\end{equation}}
Similarly, substituting \eqref{eq:state observer} into the derivative of \eqref{eq:e2} yields the trajectory tracking error dynamics
\begin{equation}
\dot{e}_2 = 
\begin{cases}
f_{\hat{x}} + u + v_{r,\sigma} - \dot{\bar{x}}_d+ Y_{\hat{x}}\hat{\theta}, & \forall t \in \left[t_\sigma^a, t_{\sigma}^u\right) , \\
f_{\hat{x}} + u - \dot{\bar{x}}_d + Y_{\hat{x}}\hat{\theta}, & \forall t \in \left[t_\sigma^u, t_{\sigma+1}^a\right).
\end{cases}
\end{equation}
Using the controller designed in \eqref{eq:controller}, the error dynamics for \( e_2 \) for both intervals can be simplified to  
\begin{equation}\label{eq:e2_dot}
\dot{e}_2 = -k_2 e_2, \quad \forall t \in \left[t_\sigma^a, t_{\sigma+1}^a\right). 
\end{equation}  
The following section presents the main theoretical results of this paper and analyzes the stability of the switched error system under the controller in \eqref{eq:controller}.

\section{Stability Analysis for Switched Systems}\label{section:x_stability}
In this section, the stability of the developed switched observer is analyzed using the controller designed in \eqref{eq:controller}. Using the Lyapunov function introduced in \eqref{eq:lyapunov} along with the user-defined bounds \( V_u \) and \( V_l \) and the error dynamics in \eqref{eq:e1_dot} and \eqref{eq:e2_dot}, this section establishes the dwell-time conditions necessary to ensure the condition $V_l \leq V (e(t)) \leq V_u$ is satisfied $\forall t$. To facilitate the analysis, let $\mathcal{E} \subset \mathbb{R}^{2n}$ be a compact set defined as
\begin{equation}
    \mathcal{E} \coloneqq \left\{[e_{1}^\top, e_{2}^\top]^\top \in \mathbb{R}^{2n}\mid x, \hat{x} \in \mathcal{X}\right\},
\end{equation}
and let $\eta > 0$ be selected such that $B(0,\eta) \subset \mathcal{E}$, where $B(0,\eta)$ denotes the open ball of radius $\eta$ around the origin. It can be observed that whenever the errors $[e_{1}^\top, e_{2}^\top]^\top \in \mathcal{E}$, the Lipschitz bounds introduced in Assumption~\ref{assump:lipschitz} hold. The following theorem presents the main result of this paper.
\begin{theorem}\label{theorem:theorem2}
Given the error system in \eqref{eq:e1_dot} and \eqref{eq:e2_dot} with initial conditions satisfying $e(t_{0}^{a}) \in \mathcal{E}$, the bounds $V_l \leq V(e(t)) \leq V_u$ hold if Assumptions \ref{assump:lipschitz}--\ref{assump:fe} hold, the switching signal satisfies the dwell-time conditions
\begin{equation}\label{eq:dwell_time_a}
    \Delta t_\sigma^a \geq -\frac{1}{\underline{k}_a} 
    \ln \left( \frac{V_l}
    {V(e(t_\sigma^a))} \right),
\end{equation}
% \begin{equation}\label{eq:dwell_time_u}
% \Delta t_\sigma^u \leq \frac{1}{\underline{k}} \ln\left( \frac{V_u + \frac{(\bar{d}+\bar{Y} \tilde{\theta}(t_\sigma^u))^2}{2 L_1{\underline{k}}}}{V(t_\sigma^u) + \frac{(\bar{d}+\bar{Y} \tilde{\theta}(t_\sigma^u))^2}{2 L_1\underline{k}}} \right).
% \end{equation}
\begin{equation}\label{eq:dwell_time_u}
\Delta t_\sigma^u \leq \frac{1}{\underline{k}_u} \ln\left( \frac{V_u +\frac{(\bar{d} + \bar{Y}\tilde{\theta}_{\sigma})^2}{2 L_1 \underline{k}_u}}{V(\Delta t_{\sigma-1}^a) + \frac{(\bar{d} + \bar{Y}\tilde{\theta}_{\sigma})^2}{2 L_1 \underline{k}_u}} \right),
\end{equation}
where \(\underline{k}_a > 0\), \(\underline{k}_u > 0\) and \(L_1 > 0 \) are user-defined gains, and the upper bound $V_u$ is selected such that $V_u < \frac{1}{2}\eta^2$.
\end{theorem}
\begin{proof} The generalized derivative $\dot{\tilde V}$ (see \cite[Equation 13]{SCC.Shevitz.Paden1994}) of the candidate Lyapunov function in \eqref{eq:lyapunov} along the flow of \eqref{eq:e1_dot} and \eqref{eq:e2_dot} is given by
\begin{equation}
\scalebox{0.95}{$\dot{\tilde V}(t,e) =
\begin{cases}
\begin{aligned}
    & e_1^\top \mathcal{K}\left( f_x - f_{\hat{x}} - v_{r,\sigma} + d_x + Y_x\theta - Y_{\hat{x}}\hat{\theta} \right) \\
    & \quad \quad + e_2^\top (-k_2 e_2), \quad \forall t \in \left[t_\sigma^a, t_{\sigma}^u\right), \\
\end{aligned} \\
\begin{aligned}
    & e_1^\top \left( f_x - f_{\hat{x}} + d_x+ Y_x\theta - Y_{\hat{x}}\hat{\theta} \right) \\
    & \quad \quad + e_2^\top (-k_2 e_2), \quad \forall t \in \left[t_\sigma^u, t_{\sigma+1}^a\right),
\end{aligned}
\end{cases}$}
\label{eq:lyuanoov_first_deri}
\end{equation}
where $\mathcal{K}$ denotes Filippov's differential inclusion \cite[Definition 2.1]{SCC.Shevitz.Paden1994}. 
\subsubsection{Analysis in the GPS-available interval}
Consider the interval $t \in \left[t_\sigma^a, t_{\sigma}^u\right)$, such that $e(t_\sigma^a) \in \mathcal{E}$. Since $e_1$ is measurable, the term \( v_{r,\sigma} \) can be implemented. Substituting \( v_{r,\sigma} \) into the first case of \eqref{eq:lyuanoov_first_deri} yields
\begin{multline}
\dot{\tilde V}(t,e) \subset  e_1^\top \Big( f_x + Y_x\theta + d_x - f_{\hat{x}} - Y_{\hat{x}}\hat{\theta} - k_1 e_1  \Big) \\
                - e_1^\top (\bar{d} + \bar{Y}\tilde{\theta}_{\sigma}) \, \text{SGN}(e_1) - e_2^\top k_2 e_2,
\end{multline}
where for a given $x$, the function $\text{SGN}(x) = 1$ if $x>1$,  $\text{SGN}(x) = -1$ if $x<-1$, and  $\text{SGN}(x) = [-1,1]$ if $x=0$. Since $f$ is Lipschitz continuous uniformly in $t$, \( \| d (t,x) \| \leq \bar{d} \), and $\|Y_x\tilde{\theta}\| \leq \bar{Y}\tilde{\theta}_{\sigma}$, for all $t \in \left[t_\sigma^a, t_{\sigma}^u\right)$ and $x \in \mathcal{X}$. Thus, $\dot{\tilde V}$ can be bounded on $\left[t_\sigma^a, t_{\sigma}^u\right) \times \mathcal{E}$ as
\begin{equation}\label{eq:V_diff_eqn_a}
\dot{\tilde V}(t,e) \leq (L_f - \lambda_{\min}(k_1)) \| e_1 \|^2  - \lambda_{\min}(k_2) \| e_2 \|^2,
\end{equation}
where \( \lambda_{\min}(k_1) \) and \( \lambda_{\min}(k_2) \) denote the minimum eigenvalues of \( k_1 \) and \( k_2 \), respectively, and for a set $A\subset \mathbb{R}$ and a scalar $b\in \mathbb{R}$, the notation $A\leq b$ is used to imply that for all $a\in A$, $a < b$.

Letting $\underline{k}_1 = -L_f +\lambda_{\min}(k_1)$ and  $\underline{k}_2 = -\lambda_{\min}(k_2)$, we can simplify \eqref{eq:V_diff_eqn_a} as
\begin{equation}
\dot{\tilde V}(t,e) \leq -\underline{k}_1 \| e_1 \|^2 - \underline{k}_2 \| e_2 \|^2,
\end{equation}
for all $(t,e) \in \left[t_\sigma^a, t_{\sigma}^u\right) \times \mathcal{E}$. Letting $\underline{k}_a = \min \{\frac{\underline{k}_1}{2},\frac{\underline{k}_2}{2} \}$, the generalized derivative of the Lyapunov function $V$ is bounded as 
\begin{equation}\label{eq:Vdot_final}
    \dot{\tilde V}(t,e) \leq -\underline{k}_a \ V(e), \quad \forall (t,e) \in \left[t_\sigma^a, t_{\sigma}^u\right) \times \mathcal{E}.
\end{equation}
Invoking, for example, \cite[Theorem~7.2]{SCC.Kamalapurkar.Dixon.ea2020}, it can be concluded that the error system in \eqref{eq:e1_dot} and \eqref{eq:e2_dot} is strongly locally asymptotically stable in $\Delta t_{\sigma}^a$.

Using the chain rule for generalized derivatives \cite[Proposition 4.2]{SCC.Kamalapurkar.Dixon.ea2020} and the comparison lemma \cite[Lemma 3.4]{SCC.Khalil1996}, the candidate Lyapunov function satisfies
\begin{equation}\label{eq:Vexpbound}
V(e(t)) \leq V(e(t_\sigma^a))e^{-\underline{k}_a (t-t_\sigma^a)}
, \quad \forall t \in \left[t_\sigma^a, t_{\sigma}^u\right).
\end{equation}
Setting \( t = t_\sigma^u \) and imposing the bound \( V(e(t_\sigma^u)) \geq V_l \), we obtain
\begin{equation}
    V_l \leq V(e(t_\sigma^u)) \leq V(e(t_\sigma^a)) e^{-\underline{k}_a (t_\sigma^u - t_\sigma^a)}.
\end{equation}
Solving for the dwell-time \( \Delta t_\sigma^a = t_\sigma^u - t_\sigma^a \), we obtain the dwell-time condition in \eqref{eq:dwell_time_a}. 
Thus, the lower bound $V_l \leq V(t_\sigma^u)$ holds provided the dwell-time condition in \eqref{eq:dwell_time_a} is satisfied.

\subsubsection{Analysis during the GPS-Denied Interval}
For \( t \in \left[t_\sigma^u, t_{\sigma+1}^a\right)\), the sliding mode term \( v_{r,\sigma} \) is not measurable. Thus, using the second case of equation \eqref{eq:lyuanoov_first_deri}, we can express the derivative of $V$ as
\begin{multline}
\dot{V}(t,e) = 
     e_1^\top \left( f_x - f_{\hat{x}} + d_x+ Y_x\theta - Y_{\hat{x}}\hat{\theta} \right) \\
     + e_2^\top (-k_2 e_2), \quad \forall t \in \left[t_\sigma^u, t_{\sigma+1}^a\right).
\end{multline}
Note that using \eqref{eq:Vexpbound} and the fact that $e(t_\sigma^a)\in \mathcal{E}$, it can be concluded that $e(t_\sigma^u) \in \overset{\circ}{\mathcal{E}}$, where $\overset{\circ}{\mathcal{E}}$ denotes the interior of $\mathcal{E}$. Let $t_\sigma^\prime$ be a time instance such that $ t_\sigma^\prime = \inf_{t > t_\sigma^u} \{e(t)\notin \mathcal{E}\}$. Since $e(t_\sigma^u) \in \overset{\circ}{\mathcal{E}}$, the interval $[t_\sigma^u,t_\sigma^\prime]$ is of nonzero length. Since $e(t)\in \mathcal{E}$ on the interval $[t_\sigma^u,t_\sigma^\prime]$, using the local Lipschitz continuity of \( f \) and \( Y\) and using the bound on \( d\), the time-derivative of the candidate Lyapunov function can be bounded on $[t_\sigma^u,t_\sigma^\prime]$ as
\begin{equation}
\dot{V}(t,e) \leq L_1 \|e_1\|^2 + (\bar{d} + \bar{Y}\tilde{\theta}_{\sigma}) \|e_1\| - \underline{k}_2 \|e_2\|^2,
\end{equation}
%for all $(t,e) \in \left[t_\sigma^u, t_{\sigma+1}^a\right) \times \mathcal{E}$}, 
where $L_1 \coloneqq L_f + L_Y \tilde{\theta}_{\sigma} $ and \( \underline{k}_2 = \lambda_{\min}(k_2) \).
Using completion of the squares
\begin{equation}
\dot{V}(t,e(t)) \leq \underline{k}_u V (e(t)) + \frac{(\bar{d} + \bar{Y}\tilde{\theta}_{\sigma})^2}{2 L_1},
\end{equation}
where $\underline{k}_u \coloneqq \min\{\frac{3L_1}{2}, \underline{k}_2\}$. 
Using the comparison lemma \cite[Lemma 3.4]{SCC.Khalil2002}, any solution to this differential inequality starting at time \( t_\sigma^u \) satisfies
\begin{equation}
V(e(t)) \leq V(e(t_\sigma^u)) e^{\underline{k}_u (t - t_\sigma^u)} + \frac{(\bar{d} + \bar{Y}\tilde{\theta}_{\sigma})^2}{2 L_1 \underline{k}_u} (e^{\underline{k}_u(t - t_\sigma^u)} - 1),
\end{equation}
$\forall t \in [t_\sigma^u,t_\sigma^\prime]$. 
Under the dwell-time condition
\begin{equation}
\Delta t_\sigma^u \leq \frac{1}{\underline{k}_u} \ln\left( \frac{V_u +\frac{(\bar{d} + \bar{Y}\tilde{\theta}_{\sigma})^2}{2 L_1 \underline{k}_u}}{V(e(t_\sigma^u)) + \frac{(\bar{d} + \bar{Y}\tilde{\theta}_{\sigma})^2}{2 L_1 \underline{k}_u}} \right),
\end{equation}
we obtain the bound 
\begin{equation}
V_u \geq V(e(t_\sigma^u)) e^{\underline{k} \Delta t_\sigma^u} + \frac{(\bar{d} + \bar{Y}\tilde{\theta}_{\sigma})^2}{2 L_1 \underline{k}} (e^{\underline{k} \Delta t_\sigma^u} - 1).
\end{equation}
Since $V_u$ is selected to ensure that $V(e) \leq V_u \implies e\in\mathcal{E}$, we conclude that $t_\sigma^\prime \geq t^a_{\sigma+1}$, and as a result, for all $t\in [t_\sigma^u,t_{\sigma+1}^a]$, \( V(e(t)) \leq V_u \) and $e(t) \in \mathcal{E}$. In particular, \( V(e(t^a_{\sigma+1})) \leq V_u \) and $e(t^a_{\sigma+1}) \in \mathcal{E}$. 

Since $e(t_0^a) \in \mathcal{E}$, an inductive argument can be used to conclude that $\forall \sigma \in \mathcal{N}^{o}$, $e(t^a_{\sigma}) \in \mathcal{E}$ and for all $t\geq t_0^a$, $V_l \leq V(e(t))\leq V_u$. 
\end{proof}

From \eqref{eq:recursive_product_subset}, we conclude that the term \( \tilde{\theta}_{\sigma} \) decreases over successive feedback-denied intervals, i.e.,
\begin{equation}\label{eq:dec_theta}
\tilde{\theta}(t_1^u) > \tilde{\theta}(t_2^u) > \cdots > \tilde{\theta}(t_{\bar{\mathcal{|N|}}}^u),
\end{equation}
where \( \bar{\mathcal{|N|}} \) denotes the switching instance when $\tilde{\theta}$ reaches the UUB bound derived in Theorem 1. Using \eqref{eq:dec_theta} and the dwell-time conditions in \eqref{eq:dwell_time_u}, the following inequality can be satisfied.
\begin{equation}
    \Delta t_1^u < \Delta t_2^u < \dots < \Delta t_{|\mathcal{\bar{N}}|}^u.
\end{equation}
That is, the time spent in the GPS-denied region can be increased with improving estimates of $\theta$, up to a limit. This completes the theoretical analysis of the switching observer and control framework under intermittent GPS availability. In the following section, simulation results are presented to validate the developed approach.

\section{Simulation}
To demonstrate the performance of the developed method, consider the following nonlinear second-order dynamical system 
\begin{equation}
    \dot{x} = \begin{bmatrix} 0 & 0 \\ -x_1 & -x_1^3 \end{bmatrix}\theta + \begin{bmatrix} x_2 \\ x_1 \end{bmatrix} +u,
\end{equation}
where \( x : [0,\infty) \mapsto \mathbb{R}^2 \), \( u : [0,\infty) \mapsto \mathbb{R}^2 \), and the unknown parameters are selected as
\begin{equation}
    \theta = \begin{bmatrix} 1 & 0.5 \end{bmatrix}^\top.
\end{equation}
The desired trajectory is selected as
\begin{equation}
    x_d(t) = \begin{bmatrix} \sin(2t) \\ 2\cos(2t) \end{bmatrix}, \quad     \dot{x}_d(t) = \begin{bmatrix} 2\cos(2t) \\ -4\sin(2t) \end{bmatrix}.
\end{equation}
The system is subjected to a random additive disturbance as unmodeled dynamics, uniformly distributed in the interval \( [-1.5,\, 1.5] \), with a known upper bound of \( \bar{d} = 1.5 \). At $t=0$, GPS signal is available and the initial conditions are
\[
x(0) = \begin{bmatrix} -1 \\ 1 \end{bmatrix}, \quad 
\hat{x}(0) = \begin{bmatrix} 0 \\ 0 \end{bmatrix}, \quad 
\hat{\theta}(0) = \begin{bmatrix} 0 \\ 0 \end{bmatrix}.
\]
The parameter estimation law employs a CL strategy that maintains a history stack of \( N = 20 \) data points during 
$\Delta t_{\sigma}^a$. The adaptation gain is chosen as \( \Gamma = 4 I_2\) and $k_{\theta} = 5$. The integral terms in \eqref{eq:filteredSignal1}--\eqref{eq:filteredSignal3}, are accumulated over a fixed time window of \( \Delta t = 0.25\,\text{s} \), and a new data point is added to the history stack when the minimum eigenvalue of the associated information matrix improves by more than a threshold \( \underline{\lambda} = 0.04 \). The delay differential equations (DDEs) governing the CL dynamics are integrated using the MATLAB \texttt{dde23} solver.

At the culmination of $[t_{0}^a,t_{0}^u)$, a switching condition is triggered as state feedback becomes unavailable and the parameter updates are provided to the agent along with the computed $\Delta t_0^u$. During \( [t_{0}^u,t_{1}^a) \), the parameter estimates are frozen (\( \dot{\hat{\theta}} = 0 \)) and the system dynamics are propagated using the MATLAB \texttt{ODE45} solver based on the current state estimates and last available parameter estimates. The simulation then proceeds to the next interval of $[t_{1}^a,t_{1}^u)$, and the process is repeated.

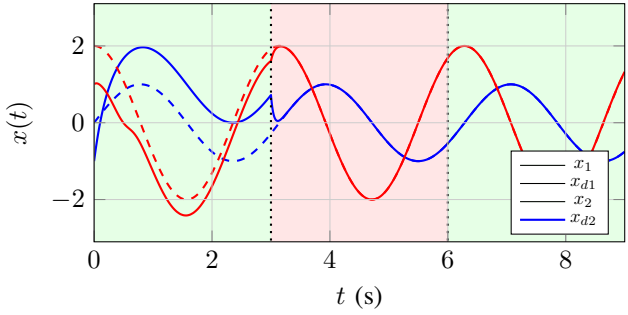
\begin{figure}
    \centering
    \input{figures/x_state}
    \caption{Comparison of the actual state trajectory \( x(t) \) with the desired trajectory \( x_d(t) \). Time intervals with GPS availability are highlighted in green, while GPS-denied intervals are shown in red. The degraded performance of the tracking controller is due to poor estimates of $\hat{\theta}$, however after the CL update the tracking performance is improved in first denied interval. }
    \label{fig:x_state_plot}
\end{figure}

\begin{figure}
    \centering
    \input{figures/x_hat}
    \caption{Comparison of the estimated state trajectory $\hat{x}(t)$ with the desired trajectory $x_d(t)$.}
    \label{fig:x_hat_plot}
\end{figure}
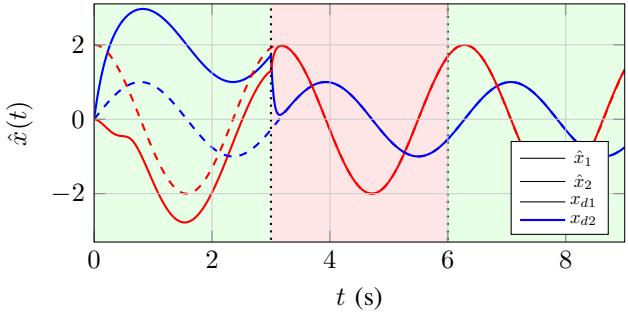

\begin{figure}
    \centering
    \input{figures/theta_tilde}
    \caption{Time evolution of the parameter estimation error $\tilde{\theta}(t)$. The error components $\tilde{\theta}_1$ and $\tilde{\theta}_2$ converge toward a neighborhood of the origin during GPS-available intervals and remain constant during GPS-denied intervals.}
    \label{fig:theta_tilde_plot}
\end{figure}
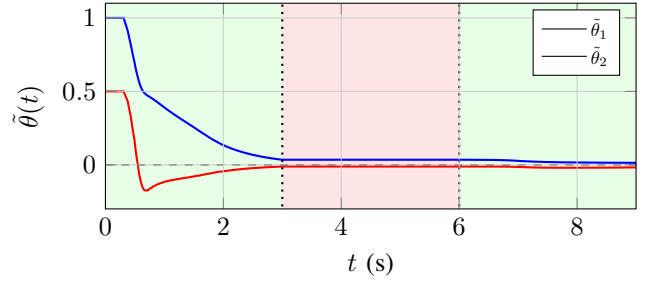

\begin{figure}
    \centering
    \input{figures/e1}
    \caption{The components of state estimation error, $(e_1)_1$ and $(e_1)_2$ are shown over time.}
    \label{fig:e1_error_plot}
\end{figure}
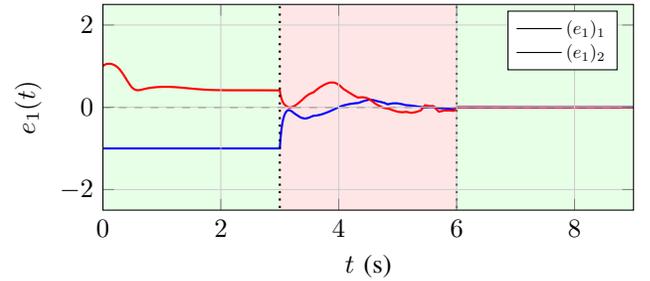

\begin{figure}
    \centering
    \input{figures/e2}
    \caption{The components of trajectory tracking error, $(e_2)_1$ and $(e_2)_2$ are shown over time.}
    \label{fig:e2_error_plot}
\end{figure}
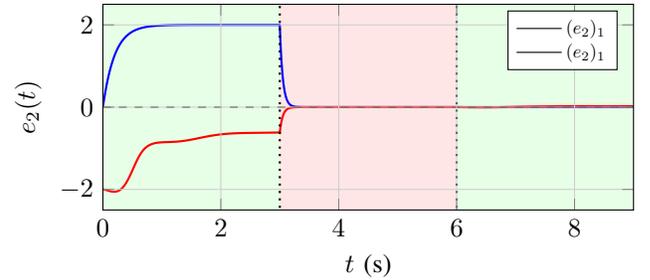

\begin{figure}
    \centering
    \input{figures/dwell}
    \caption{Variation of the dwell-time $\Delta t_{\sigma}^u$ across successive switching instances $\sigma$, illustrating the allowable time at start of each GPS-denied interval.}
    \label{fig:dwell_time_plot}
\end{figure}
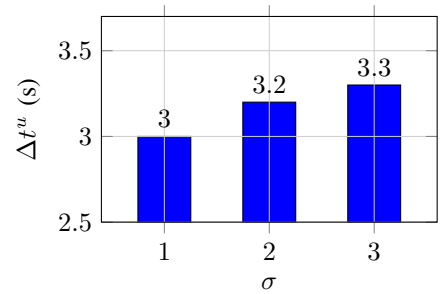

\section{Discussion}

Figure \ref{fig:x_state_plot} shows the trajectory of the actual system state \( x \) alongside the desired trajectory \( x_d \), illustrating successful convergence and tracking. The estimated state trajectory \( \hat{x}(t) \), shown in Figure \ref{fig:x_hat_plot}, closely follows the desired path despite switching between GPS-available and GPS-denied intervals. The evolution of the parameter estimation error \( \tilde{\theta}(t) \) is depicted in Figure \ref{fig:theta_tilde_plot}, indicating convergence of the estimates to actual parameters over time. Figures \ref{fig:e1_error_plot} and \ref{fig:e2_error_plot} present the state estimation error \( e_1(t) \) and the trajectory tracking error \( e_2(t) \), respectively, both of which remain bounded and diminish over time, thereby demonstrating the effectiveness of the integral concurrent learning mechanism under the switching feedback conditions. Figure \ref{fig:dwell_time_plot} shows that as the parameters are learned more accurately and the observer provides better state estimates, the time the agent can operate in the GPS-denied region increases. These refined parameter estimates allow for longer periods of autonomous navigation without GPS, significantly extending the agent's operational range without continuous state feedback.

\section{Conclusion}
In a switched system framework, the given error bounds impose the dwell-time constraints during operation in GPS-denied intervals. To address this constraint, this paper integrates a switched systems framework with memory regressor extension techniques. By leveraging state measurements in GPS-available intervals the unmodeled modeled parameters can be estimated. By compensating for unmodeled dynamics during operation in GPS-denied intervals, the method effectively controls the growth of state estimation errors and enables prolonged operation in GPS-denied environments.

Future work will consider more general dynamic models that extend beyond the static parametric structure. Additionally, we aim to generalize the developed framework to spatially complex environments where GPS availability is topology-dependent, necessitating the integration of trajectory planning techniques while relaxing dwell-time constraints.

\bibliographystyle{IEEETrans.bst}
\bibliography{scc,sccmaster,muzaffar}
\end{document}

%% file: figures/x_state.tex
\begin{tikzpicture}
\begin{axis}[
    xlabel={$t$ (s)},
    ylabel={$x(t)$},
    legend pos=south east,
    legend style={nodes={scale=0.7, transform shape}},
    width=\linewidth,
    height=0.55\linewidth,
    grid=both,
        xmin = 0,
        xmax = 9,
    ymin = -3.1,
    ymax = 3.1,
    axis on top=true,            % Bring axes + grid to foreground
    grid=both,
    grid style={line width=0.3pt, draw=gray!40}, % light gray grid
    clip=true                    % Ensure fills don’t spill
]
% --- Background fill: t = 0 to 3 ---
\addplot [
    draw=none,
    fill=green!10,   % lighter green
    opacity=1
] coordinates {(0, -10) (3, -10) (3, 10) (0, 10)};

% --- Background fill: t = 3 to 6 ---
\addplot [
    draw=none,
    fill=red!10,     % lighter red
    opacity=1
] coordinates {(3, -10) (6, -10) (6, 10) (3, 10)};

% --- Background fill: t = 6 to 9 ---
\addplot [
    draw=none,
    fill=green!10,   % lighter green
    opacity=1
] coordinates {(6, -10) (9, -10) (9, 10) (6, 10)};
% x1(t)
\addplot[blue, thick] table [x index=0, y index=1] {data/xDataSim1.dat};
\addlegendentry{$x_1$}

% x1_desired(t)
\addplot[blue, dashed, thick] table [x index=0, y index=1] {data/xDesiredSim1.dat};
\addlegendentry{$x_{d1}$}

% x2(t)
\addplot[red, thick] table [x index=0, y index=2] {data/xDataSim1.dat};
\addlegendentry{$x_2$}

% x2_desired(t)
\addplot[red, dashed, thick] table [x index=0, y index=2] {data/xDesiredSim1.dat};
\addlegendentry{$x_{d2}$}

% --- Vertical dotted lines at t=3 and t=6 ---
\addplot[black, dotted, thick] coordinates {(3, -10) (3, 10)};
\addplot[black, dotted, thick] coordinates {(6, -10) (6, 10)};

\end{axis}
\end{tikzpicture}

%% file: figures/x_hat.tex
\begin{tikzpicture}
\begin{axis}[
    xlabel={$t$ (s)},
    ylabel={$\hat{x}(t)$},
    legend pos=south east,
    legend style={nodes={scale=0.7, transform shape}},
    width=\linewidth,
    height=0.55\linewidth,
    grid=both,
    xmin = 0,
    xmax = 9,
    ymin = -3.3,
    ymax = 3.1,
    axis on top=true,            % Bring axes + grid to foreground
    grid=both,
    grid style={line width=0.3pt, draw=gray!40}, % light gray grid
    clip=true                    % Ensure fills don’t spill
]

% --- Background fill: t = 0 to 3 ---
\addplot [
    draw=none,
    fill=green!10,   % lighter green
    opacity=1
] coordinates {(0, -10) (3, -10) (3, 10) (0, 10)};

% --- Background fill: t = 3 to 6 ---
\addplot [
    draw=none,
    fill=red!10,     % lighter red
    opacity=1
] coordinates {(3, -10) (6, -10) (6, 10) (3, 10)};

% --- Background fill: t = 6 to 9 ---
\addplot [
    draw=none,
    fill=green!10,   % lighter green
    opacity=1
] coordinates {(6, -10) (9, -10) (9, 10) (6, 10)};
% --- xhat_1(t) ---
\addplot[blue, thick] table [x index=0, y index=1] {data/xHatDataSim1.dat};
\addlegendentry{$\hat{x}_1$}

% --- xhat_2(t) ---
\addplot[red, thick] table [x index=0, y index=2] {data/xHatDataSim1.dat};
\addlegendentry{$\hat{x}_2$}

% --- x_desired_1(t) ---
\addplot[blue, dashed, thick] table [x index=0, y index=1] {data/xDesiredSim1.dat};
\addlegendentry{$x_{d1}$}

% --- x_desired_2(t) ---
\addplot[red, dashed, thick] table [x index=0, y index=2] {data/xDesiredSim1.dat};
\addlegendentry{$x_{d2}$}

% --- Vertical dotted lines at t = 3 and t = 6 ---
\addplot[black, dotted, thick] coordinates {(3, -10) (3, 10)};
\addplot[black, dotted, thick] coordinates {(6, -10) (6, 10)};
\end{axis}
\end{tikzpicture}

%% file: figures/theta_tilde.tex
\begin{tikzpicture}
\begin{axis}[
    xlabel={$t$ (s)},
    ylabel={$\tilde{\theta}(t)$},
    legend pos=north east,
    legend style={nodes={scale=0.7, transform shape}},
    width=\linewidth,
    height=0.5\linewidth,
    grid=both,
        xmin = 0,
        xmax = 9,
        ymin = -0.3,
        ymax = 1.1,
    axis on top=true,            % Bring axes + grid to foreground
    grid=both,
    grid style={line width=0.3pt, draw=gray!40}, % light gray grid
    clip=true                    % Ensure fills don’t spill
]
% --- Background fill: t = 0 to 3 ---
\addplot [
    draw=none,
    fill=green!10,   % lighter green
    opacity=1
] coordinates {(0, -10) (3, -10) (3, 10) (0, 10)};

% --- Background fill: t = 3 to 6 ---
\addplot [
    draw=none,
    fill=red!10,     % lighter red
    opacity=1
] coordinates {(3, -10) (6, -10) (6, 10) (3, 10)};

% --- Background fill: t = 6 to 9 ---
\addplot [
    draw=none,
    fill=green!10,   % lighter green
    opacity=1
] coordinates {(6, -10) (9, -10) (9, 10) (6, 10)};

% theta_tilde_1(t)
\addplot[blue, thick] table [x index=0, y index=1] {data/thetaTildeSim1.dat};
\addlegendentry{$\tilde{\theta}_1$}

% theta_tilde_2(t)
\addplot[red, thick] table [x index=0, y index=2] {data/thetaTildeSim1.dat};
\addlegendentry{$\tilde{\theta}_2$}

% Zero line
\addplot[black, dashed, domain=0:10] {0};

% --- Vertical dotted lines at t=3 and t=6 ---
\addplot[black, dotted, thick] coordinates {(3, -10) (3, 10)};
\addplot[black, dotted, thick] coordinates {(6, -10) (6, 10)};

\end{axis}
\end{tikzpicture}

%% file: figures/e1.tex
\begin{tikzpicture}
\begin{axis}[
    xlabel={$t$ (s)},
    ylabel={$e_1(t)$},
    legend pos=north east,
    legend style={nodes={scale=0.7, transform shape}},
    width=\linewidth,
    height=0.5\linewidth,
    xmin = 0,
    xmax = 9,
    ymin = -2.5,
    ymax = 2.5,
    axis on top=true,            % Bring axes + grid to foreground
    grid=both,
    grid style={line width=0.3pt, draw=gray!40}, % light gray grid
    clip=true                    % Ensure fills don’t spill
]

% --- Background fill: t = 0 to 3 ---
\addplot [
    draw=none,
    fill=green!10,   % lighter green
    opacity=1
] coordinates {(0, -10) (3, -10) (3, 10) (0, 10)};

% --- Background fill: t = 3 to 6 ---
\addplot [
    draw=none,
    fill=red!10,     % lighter red
    opacity=1
] coordinates {(3, -10) (6, -10) (6, 10) (3, 10)};

% --- Background fill: t = 6 to 9 ---
\addplot [
    draw=none,
    fill=green!10,   % lighter green
    opacity=1
] coordinates {(6, -10) (9, -10) (9, 10) (6, 10)};

% --- Blue trajectory ---
\addplot[blue, thick] table [x index=0, y index=1] {data/e1Sim1.dat};
\addlegendentry{$(e_1)_1$}

% --- Red trajectory ---
\addplot[red, thick] table [x index=0, y index=2] {data/e1Sim1.dat};
\addlegendentry{$(e_1)_2$}

% --- Zero reference line ---
\addplot[black, dashed, domain=0:10] {0};

% --- Vertical dotted lines ---
\addplot[black, dotted, thick] coordinates {(3, -10) (3, 10)};
\addplot[black, dotted, thick] coordinates {(6, -10) (6, 10)};

\end{axis}
\end{tikzpicture}

%% file: figures/e2.tex
\begin{tikzpicture}
\begin{axis}[
    xlabel={$t$ (s)},
    ylabel={$e_2(t)$},
    legend pos=north east,
    legend style={nodes={scale=0.7, transform shape}},
    width=\linewidth,
    height=0.5\linewidth,
    grid=both,
        xmin = 0,
        xmax = 9,
        ymin = -2.5,
        ymax = 2.5,
    axis on top=true,            % Bring axes + grid to foreground
    grid=both,
    grid style={line width=0.3pt, draw=gray!40}, % light gray grid
    clip=true                    % Ensure fills don’t spill
]

% --- Background fill: t = 0 to 3 ---
\addplot [
    draw=none,
    fill=green!10,   % lighter green
    opacity=1
] coordinates {(0, -10) (3, -10) (3, 10) (0, 10)};

% --- Background fill: t = 3 to 6 ---
\addplot [
    draw=none,
    fill=red!10,     % lighter red
    opacity=1
] coordinates {(3, -10) (6, -10) (6, 10) (3, 10)};

% --- Background fill: t = 6 to 9 ---
\addplot [
    draw=none,
    fill=green!10,   % lighter green
    opacity=1
] coordinates {(6, -10) (9, -10) (9, 10) (6, 10)};

% theta_tilde_1(t)
\addplot[blue, thick] table [x index=0, y index=1] {data/e2Sim1.dat};
\addlegendentry{$(e_2)_1$}

% theta_tilde_2(t)
\addplot[red, thick] table [x index=0, y index=2] {data/e2Sim1.dat};
\addlegendentry{$(e_2)_1$}

% Zero line
\addplot[black, dashed, domain=0:10] {0};

% --- Vertical dotted lines at t=3 and t=6 ---
\addplot[black, dotted, thick] coordinates {(3, -10) (3, 10)};
\addplot[black, dotted, thick] coordinates {(6, -10) (6, 10)};

\end{axis}
\end{tikzpicture}

%% file: figures/dwell.tex
\begin{tikzpicture}
\begin{axis}[
    ybar,
    bar width=20pt,
    enlarge x limits=0.3,
    xlabel={$\sigma$},
    ylabel={$\Delta t^u$ (s)},
    ymin=2.5,
    ymax=3.7,
    xtick={1,2,3},
    xticklabels={$1$,$2$,$3$},
    grid=both,
    width=0.7\linewidth,
    height=0.5\linewidth,
    nodes near coords,
    nodes near coords align={vertical},
    legend style={at={(0.5,-0.15)}, anchor=north, legend columns=1},
    axis on top=true,            % Bring axes + grid to foreground
    grid=both,
    grid style={line width=0.3pt, draw=gray!40}, % light gray grid
    clip=true                    % Ensure fills don’t spill
]

% Dwell time data: Example values for sigma = 1, 2, 3
\addplot[fill=blue] coordinates {(1, 3) (2, 3.2) (3, 3.3)};

\end{axis}
\end{tikzpicture}